\documentclass[sigconf]{acmart}
\setcopyright{none} 
\acmConference[Anonymous Submission to WWW 2018]{The Web Conference}{Due Oct 31 2017}{Lyon, France}
\acmYear{2018}

\usepackage{booktabs}

\usepackage{epsfig,amsmath,multirow,makecell,caption,soul,csquotes,color,wrapfig,subcaption,mathtools,bm,upgreek,spverbatim}

\newenvironment{proof}{\paragraph{Proof:}}{\hfill$\square$}

\usepackage[boxed, ruled, vlined, linesnumbered]{algorithm2e}

\DeclareMathOperator*{\argmin}{arg\,min}

\makeatletter
\providecommand{\leadsfrom}{%
  \mathrel{\mathpalette\reflect@squig\relax}%
}
\newcommand{\reflect@squig}[2]{%
  \reflectbox{$\m@th#1\leadsto$}%
}
\makeatother

\usepackage{graphicx}

\newenvironment{changemargin} [2]{\begin{list}{}{
          \setlength{\topsep}{0pt}\setlength{\leftmargin}{0pt}
          \setlength{\rightmargin}{0pt}
          \setlength{\listparindent}{\parindent}
          \setlength{\itemindent}{\parindent}
          \setlength{\parsep}{0pt plus 1pt}
          \addtolength{\leftmargin}{#1}\addtolength{\rightmargin}{#2}
          }\item }{\end{list}}

\newenvironment{myitemize} 
   {
     \begin{changemargin}{-3pt}{-0cm}
     \vspace{-10pt}
     \hspace{-5pt}
     \begin{itemize}
     \setlength{\itemsep}{3pt}
   }
   {
     \end{itemize}
     \vspace{-0pt}
     \end{changemargin}
   }

   {
     \begin{changemargin}{-8pt}{-0cm}
     \vspace{-13pt}
     \hspace{5pt}
     \begin{description}
     \setlength{\itemsep}{-1pt}
   }
   {
     \end{description}
     \end{changemargin}
   }

   {
     \begin{changemargin}{-8pt}{-0cm}
     \vspace{-13pt}
     \hspace{5pt}
     \begin{enumerate}
     \setlength{\itemsep}{1pt}
   }
   {
     \end{enumerate}
     \end{changemargin}
   }

\newcommand{\system}{{\sf dp-GAN}\xspace}

\newcommand{\gan}{GAN\xspace}

\begin{document}

\title{Differentially Private Releasing via Deep Generative Model\\
(Technical Report)}



%

\author{Xinyang Zhang}
\affiliation{%
  \institution{Lehigh University}
}
\email{xizc15@lehigh.edu}

\author{Shouling Ji}
\affiliation{%
  \institution{Zhejiang University}
}
\email{sji@gatech.edu}

\author{Ting Wang}
\affiliation{%
  \institution{Lehigh University}
}
\email{inbox.ting@gmail.com}

\maketitle



\subsection*{Abstract}
Privacy-preserving releasing of complex data (e.g., image, text, audio) represents a long-standing challenge for the data mining research community. Due to rich semantics of the data and lack of {\em a priori} knowledge about the analysis task, excessive sanitization is often necessary to ensure privacy, leading to significant loss of the data utility. In this paper, we present \system, a general private releasing framework for semantic-rich data. Instead of sanitizing and then releasing the data, the data curator publishes a deep generative model which is trained using the original data in a differentially private manner; with the generative model, the analyst is able to produce an unlimited amount of synthetic data for arbitrary analysis tasks. In contrast of alternative solutions, \system highlights a set of key features: (i) it provides theoretical privacy guarantee via enforcing the differential privacy principle; (ii) it retains desirable utility in the released model, enabling a variety of otherwise impossible analyses; and (iii) most importantly, it achieves practical training scalability and stability by employing multi-fold optimization strategies. Through extensive empirical evaluation on benchmark datasets and analyses, we validate the efficacy of \system.\\

(The source code and the data used in the paper is available at: https://github.com/alps-lab/dpgan)

\section{Introduction}
\label{sec:intro}

With the continued advances in mobile computing and the surging popularity of social media, a massive amount of semantic-rich data (e.g., image, text, audio) about individuals is being collected. While analyzing and understanding such data entails tremendous commercial value (e.g., targeted advertisements and personalized recommendations), governments and organizations all have recognized the critical need of respecting individual privacy in such practice~\cite{apple}. In general, privacy protection can be enforced in two settings. In the {\em interactive} setting, a trusted curator collects data from individuals and provides a privacy-preserving interface for the analyst to execute queries over the data; in the more challenging {\em non-interactive} setting, the curator releases a ``sanitized'' version of the data, simultaneously providing analysis utility for the analyst and privacy protection for the individuals represented in the data~\cite{Dwork:2014:book}.

Hitherto, privacy-preserving releasing of semantic-rich data still represents a long-standing challenge for the privacy and security research communities: the rich semantics of such data enable a wide variety of potential analyses, while the concrete analyses are often unknown ahead of releasing, especially in the case of exploratory data analysis. Therefore, to ensure privacy, excessive sanitization is often necessary, which may completely destroy the data utility for potential analyses.

\begin{figure}
\epsfig{file=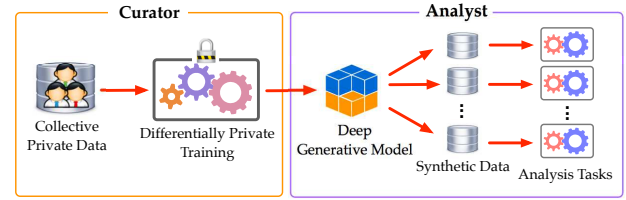, width=85mm}
\caption{High-level design of \system, a privacy-preserving releasing framework for semantic-rich data. \label{fig:framework}}
\end{figure}

In this paper, we tackle this challenge by integrating the state-of-the-art deep learning methods with advanced privacy-preserving mechanism. Specifically, we present \system, a new private releasing framework for semantic-rich data. With \system, instead of releasing a sanitized version of the original data, the curator publishes a generative model (i.e., generative adversarial network~\cite{Goodfellow:2014:nips}), which is trained using the original data in a privacy-preserving manner. The analyst, once equipped with this generative model, is able to produce synthetic data for the intended analysis tasks. The high-level framework of \system is illustrated in Figure~\ref{fig:framework}.

In comparison with alternative solutions (e.g., sanitizing and then releasing the data),  \system highlights with a number of significant advantages. First, it enforces {\em differential privacy}~\cite{Dwork:2006:icalp}, the state-of-the-art privacy principle, in the training of generative models. Due to its closure under post-processing property~\cite{Dwork:2014:book}, differential privacy ensures that the released model provides theoretically guaranteed privacy protection for the training data. Second, the use of generative models (e.g., generative adversarial networks in particular) as the vehicles of data releasing enables the synthesized data to capture the rich semantics of the original data. The faithful preservation of desirable utility leads to a variety of otherwise impossible analyses. For example, we show empirically that \system is able to effectively support semi-supervised classification tasks. Finally, the generative model is able to produce an unlimited amount of synthetic data for arbitrary analysis tasks, as shown in Figure~\ref{fig:framework}.

However, realizing \system entails two major challenges. First, it requires new algorithmic advances to implement differential privacy within generative model training. To this end, we extend the framework of Improved Wasserstein GAN~\cite{Gulrajani:2017:wganip} by integrating the state-of-the-art privacy enhancing mechanisms (e.g., Gaussian mechanism~\cite{Dwork:2014:book}) and provide refined analysis of privacy loss within this framework. Second, the stability and scalability issues of training \gan models are even more evident once privacy enhancing mechanisms are incorporated. To this end, we develop multi-fold optimization strategies, including {\em weight clustering}, {\em adaptive clipping}, and {\em warm starting}, which significantly improve both training stability and utility retention. Our contributions can be summarized as follows.

\begin{myitemize}
\item First, to our best knowledge, \system is the first working framework that realizes the paradigm of privacy-preserving model releasing for semantic-rich data. We believe this new paradigm is applicable for a broad range of privacy-sensitive data publishing applications.
\item Second, in implementing \system, we develop multi-fold system optimization strategies that not only successfully incorporate privacy enhancing mechanisms within training deep generative model, but also significantly improve the stability and scalability of generative model training itself.
\item Third, we conduct extensive empirical evaluation using real large-size image data to validate the efficacy of \system. We show that \system, besides providing theoretically guaranteed privacy protection,  preserves desirable utility of the original data, enabling a set of otherwise impossible analysis tasks.
\end{myitemize}

The remainder of the paper proceeds as follows. Section~\ref{sec:background} reviews the background of deep generative models and differential privacy; Section~\ref{sec:model} presents the high-level design of \system; Section~\ref{sec:opt} details its implementation, in particular, the multi-fold optimizations to improve the stability and scalability of model training; Section~\ref{sec:eval} empirically evaluates our proposed solution; Section~\ref{sec:liter} discusses additional relevant literature; The paper is concluded in Section~\ref{sec:end}.

\section{Preliminaries}
\label{sec:background}

In this section, we introduce the two basic building blocks of \system, generative adversarial network and differential privacy.


\begin{figure}
\centering
\epsfig{file=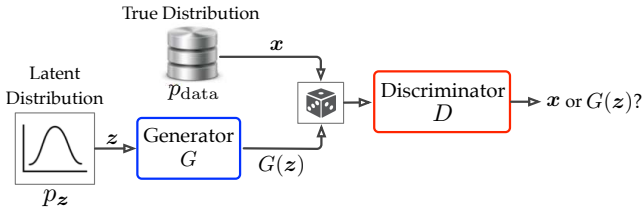, width=85mm}
\caption{Illustration of generative adversarial networks. \label{fig:gan}}
\end{figure}

\subsection{Generative Adversarial Network}
The generative adversarial network (GAN)~\cite{Goodfellow:2014:nips}
is a class of unsupervised learning algorithms which are implemented by
an adversarial process. As illustrated in Figure~\ref{fig:gan}, the \gan architecture typically comprises two neural networks, a generator $G$ and a discriminator $D$, in which $G$ learns to map from a latent distribution $p_z$ to the true data distribution $p_{\rm data}$, while $D$ discriminates between instances sampled from $p_{\rm data}$ and that generated by $G$. Here $G$'s objective is to ``fool'' $D$ by synthesizing instances that appear to have come from $p_{\rm data}$. This framework corresponds to solving a minimax two-player game with the following objective function:
\begin{equation}
\min_\theta \max_w \mathbb{E}_{x \sim p_{\text{data}}} [\log D_w(x)] + \mathbb{E}_ {z \sim p_z} [\log (1 - D_w(G_\theta(z)))]
\end{equation}
where $x$ and $z$ are sampled from $p_\text{data}$ and $p_{\bm{z}}$ respectively.

Since its advent, GAN finds applications in varied unsupervised and semi-supervised learning tasks~\cite{Chen:2016:infogan, Radford:2015:dcgan, Donahue:2016:adl, Kumar:2017:semiinv, Reed:2016:generative,Ledig:2016:superres,Yeh:2016:inpaint,Rajeswar:2017:adversarial}.
%
One line of work takes the trained discriminator as a feature
extractor and applies it in varied settings; the other line focuses on
the latent variable $z$ in the generator, either using regularization to make $z$ semantically
meaningful~\cite{Donahue:2016:adl, Chen:2016:infogan} or extracting information in the latent space
directly~\cite{Radford:2015:dcgan}.

Despite its simplicity, the original GAN formulation is
unstable and inefficient to train. A number of followup work
~\cite{Zhao:2016:energy, Chen:2016:infogan, Radford:2015:dcgan,
	Nowozin:2016:fgan,Arjovsky:2017:wgan, Gulrajani:2017:wganip}
propose new training procedures and network architectures to improve training stability and convergence rate. In particular, the Wasserstein generative adversarial network (WGAN)~\cite{Arjovsky:2017:wgan} and Improved Training of Wasserstein GANs~\cite{Gulrajani:2017:wganip}
attempt to minimize the earth mover distance between the synthesized distribution and the true distribution rather than their Jensen-Shannon divergence as in the original GAN formulation. Formally, improved WGAN adopts the following objective functions:
\begin{align}
& \argmin_w -D_w(G_\theta(z)) \\
& \argmin_\theta D_w(G_\theta(z)) - D_w(x) + \lambda \left( \left\Vert \nabla_{\hat{x}} D_w(\hat{x})\right\Vert_2 - 1\right)^2
\end{align}
Here, $\hat x = \alpha x + (1 - \alpha) G_\theta(z)$, in which $\alpha$ is
a random number sampled from $[0, 1]$. The regularization term enforces the norm of $D$'s gradients to be close to 1. This formulation is shown to allow more stable and faster training~\cite{Gulrajani:2017:wganip}.

In the following, without loss of generality, we will exemplify with the improved WGAN formulation to implement \system.

%
%
%

\subsection{Differential Privacy}

By providing theoretically guaranteed protection, differential privacy (DP)~\cite{Dwork:2009:tcc, Dwork:2006:icalp, Dwork:2014:book} is considered one of the strongest privacy definitions.

\vspace{3pt}
{\bf Definitions.}
We say a randomized mechanism $\mathcal{M} : \mathcal{D}^n \mapsto \mathcal{R}$ satisfies
  $\epsilon$-DP if for any adjacent databases
 $d, d' \in \mathbb{D}^n$ (which are identical except for one single data entry) and any subset $R \subseteq \mathcal{R}$, it holds that
 ${\rm Pr}[\mathcal{M}(d) \in R] \leq e^\epsilon {\rm Pr}[\mathcal{M}(d') \in R]$. A relaxed version, $(\epsilon, \delta)$-DP, allows the plain $\epsilon$-DP to be compromised
  with a small probability $\delta$:
 ${\rm Pr}[\mathcal{M}(d) \in R] \leq e^\epsilon {\rm Pr}[\mathcal{M}(d') \in R] + \delta$.
 In this work, we consider $(\epsilon, \delta)$-DP as the default privacy definition.

\vspace{3pt}
{\bf Mechanisms.}
For a given deterministic function $f$, DP is often achieved by injecting random noise into $f$'s output, while the noise magnitude is determined by $f$'s  sensitivity. If $f$ is vector-valued, i.e., $f:\mathcal{D}^n \mapsto \mathcal{R}^m$, its sensitivity is defined as:
 $\Updelta f = \max_{d, d'} \| f(d) - f(d') \|$, where $\Updelta f$ represents the maximum influence of a single data entry on $f$'s output, quantifying the (worst-case) uncertainty to be added to $f$'s output to hide the presence of that entry.

If $f$'s sensitivity is defined using $\ell_2$ norm, the Gaussian mechanism~\cite{Dwork:2014:book} is a common choice for randomizing $f$'s output:
\begin{equation}
	\nonumber
	\mathcal{M} (d) = f(d) +  \mathcal{N}(0, (\Updelta f)^2 \sigma \mathcal{I}),
\end{equation}
where $\mathcal{N}(0, (\Updelta f)^2 \sigma \mathcal{I})$ is a Gaussian
distribution with zero mean and covariance matrix $ (\Updelta f)^2 \sigma \mathcal{I}$ and $\mathcal{I}$ is the identity matrix.

\vspace{3pt}
{\bf Properties.} In addition, DP also features the following key properties, which we leverage in implementing \system.
\begin{myitemize}
\item {\em Closure under post-processing}. Any computation on the output of a DP-mechanism does not increase privacy loss.
\item {\em Sequential composability}. The composition of a sequence of DP-mechanisms is also DP-satisfying.
\end{myitemize}
We may use the composition theorems~\cite{Dwork:2014:book,Dwork:2010:boosting} to estimate the privacy loss after $k$-fold application of DP-mechanisms.

\section{Models and Algorithms}
\label{sec:model}

In this section, we present the basic design of \system, a generic framework for differentially private releasing of semantic-rich data.

\subsection{Overview}

Similar to the line of work on differentially private deep learning (e.g.,~\cite{Abadi:2016:dpdl}), \system achieves DP by injecting random noise in the optimization procedure (e.g., stochastic gradient descent~\cite{Song:2013:stochastic}). Yet, the GAN architecture, which comprises a generator $G$ and a discriminator $D$, presents unique challenges for realizing this idea. A na\"{i}ve solution is to inject noise in training both $G$ and $D$; the minimax game formulation however makes it difficult to tightly estimate the privacy loss, resulting in excessive degradation in the produced models.

We opt to add random perturbation only in training $D$. The rationale behind our design choice is as follows. First, as shown in Figure~\ref{fig:gan}, the real data is directly accessible only by $D$; thus, it suffices to control the privacy loss in training $D$. Second, in comparison with $G$, which often employs building blocks such as batch normalizations~\cite{Ioffe:2015:bn} and residual layers~\cite{He:2016:resnet, He:2016:imapresnet} in order to generate realistic samples, $D$ often features a simpler architecture and a smaller number of parameters, which make it possible to tightly estimate the privacy loss.

%

After deciding where to enforce privacy protection, next we present the basic construct of \system, as sketched in Algorithm \ref{alg:basic}. At a high level, \system is built upon the improved WGAN framework and enforces DP by injecting random noise in updating the discriminator $D$. Specifically, when computing $D$'s gradients with respect to a real sample $x$ (line 7), we first clip the gradients by a threshold $C$ (line 8), ensuring that the sensitivity is bounded by $C$; we then add random noise sampled from a Gaussian distribution. Additionally, we use a privacy accountant $\mathcal{A}$ similar to~\cite{McSherry:2009:PIQ:1559845.1559850} to track the cumulative privacy loss. This process iterates until convergence or exceeding the privacy
budget (line 14).

%
%


 \SetKwInput{Require}{Require}
 \begin{algorithm}


 	\caption{Basic \system}
 	 \label{alg:basic}
 	\KwIn{$n$ - number of samples; $\lambda$ - coefficient of gradient penalty; $n_\text{critic}$ - number of critic iterations
 		per generator iteration; $n_\text{param}$ - number of discriminator's parameters; $m$ - batch size; ($\alpha,
 		\beta_1, \beta_2$) - Adam hyper-parameters; $C$ - gradient clipping bound; $\sigma$ - noise scale; ($\epsilon_0$, $\delta_0$) - total privacy budget}
  \KwOut{differentially private generator $G$}
 	\While{$\theta$ has not converge} {
 		\For {$t = 1, \cdots, n_\text{critic}$} {
 			\For {$i = 1, \cdots, m$} {
 				sample $x \sim p_{\rm data}$, $z \sim p_z$, $\rho \sim \mathcal{U}\left[0, 1\right]$ \;
 				$\hat{x} \gets \rho x + (1 - \rho)  G(z) $ \;
 				$\ell^{(i)} \gets D(G (z)) - D(x) +
 				\lambda \left( \left\Vert \triangledown_{\hat{x}}  D( \hat{x} )
 				\right\Vert_2 - 1 \right)^2$ \;
                %
                \tcp{computing discriminator's gradients}
                $g^{(i)} \leftarrow \triangledown_{w} \ell^{(i)}$\;
                \tcp{clipping and perturbation ($\xi \sim \mathcal{N}\left(0, (\sigma C\right)^2\mathcal{I})$)}
                $g^{(i)} \gets g^{(i)} / \max(1, || g^{(i)} ||_2/C ) + \xi$\;
        	}
            \tcp{updating privacy accountant}
 		 				update $\mathcal{A}$ with $(\sigma, m, n_\text{param})$ \;
            \tcp{updating discriminator}
 		$w \gets \text{Adam}\left( \frac{1}{m} \sum_{i = 1}^m g^{(i)},
 		w, \alpha, \beta_1, \beta_2 \right)$\;
 	}
 			sample $\{ z^{(i)} \}_{i = 1}^m \sim p_z$ \;
            \tcp{updating generator}
            $\theta \gets \text{Adam}\left( \nabla_\theta \frac{1}{m} \sum_{i = 1}^m
 			-D( G(z^{(i)}) ) , \theta, \alpha, \beta_1, \beta_2
 			\right)  $\;
            \tcp{computing cumulative privacy loss}
 			$\delta \gets$ query $\mathcal{A}$ with $\epsilon_0$\;
            \lIf {$\delta > \delta_0$} {
 				break}}
        \Return{G}
 \end{algorithm}

\subsection{Privacy Analysis}

A key component of \system is to keep track the cumulative privacy loss during the course of training, i.e., privacy accountant $\mathcal{A}$, which integrates two building blocks:
moments accounting and sub-sampling. Next we elaborate on each component.

\subsubsection*{\bf Moments Accounting}

In~\cite{Abadi:2016:dpdl}, Abadi {\em et al.} propose moments accounting, a privacy accounting method, which provides tighter estimation of the privacy loss than the composition theorems. Specifically,
consider the privacy loss as a random
variable $Z$, which is defined as:
\begin{equation}
    \nonumber
Z(o; \mathcal{M}, d, d^\prime) = \log \frac{{\rm Pr}[\mathcal{M}(x) = o]}{
	{\rm Pr}[\mathcal{M}(d^\prime) = o]}
\end{equation}
where $d, d^\prime \in \mathcal{D}^n$ are two neighboring datasets, $\mathcal{M}$ is the random
mechanism, and $o \in \mathcal{R}$ is an outcome.

The privacy loss can be estimated by bounding the $\lambda$--th
moment of $Z$, which is calculated
via evaluating the moment generating function of $Z$ at $\lambda$:
\begin{equation}
    \nonumber
\alpha_\mathcal{M} (\lambda; d, d^\prime) =
\log \mathbb{E}_{o \sim \mathcal{M}(d)} \left[\exp(
\lambda Z(o; \mathcal{M}, d, d'))\right]
\end{equation}
To enforce DP, one needs to consider
$\alpha_\mathcal{M}$ across all possible $d, d'$, i.e.,
$\alpha_\mathcal{M} \triangleq \max_{d, d'} \alpha_\mathcal{M}
(\lambda; d, d' )$.

Using Markov's inequality, it can be proved that for any $\epsilon > 0$, $\mathcal{M}$
satisfies $(\epsilon, \delta)$-DP for
$\delta = \min_{\lambda} ( \alpha_\mathcal{M} - \lambda \epsilon )$~\cite{Abadi:2016:dpdl}.
Besides, if $\mathcal{M}$ is the composition of a sequence of
sub-mechanisms $\{\mathcal{M}_j\}_{j=1}^J$, it holds that
$\alpha_\mathcal{M}(\lambda) \leq
\sum_{j=1}^J \alpha_{\mathcal{M}_j} (\lambda)$.
In tracking the privacy loss, we apply numerical integration to compute $\alpha_\mathcal{M} (\lambda)$.

\subsubsection*{\bf Sub-sampling}
During each iteration of training $D$, we sample a batch of examples from the real dataset (line 4). The randomness due to sampling adds another level of privacy protection.
According to the privacy amplification
theorems~\cite{Beimel:2014:sample, Kasiviswanathan:2011:sample}, this sampling procedure achieves
$(\mathcal{O}(q \epsilon, q \delta))$-DP per iteration with respect to
the whole dataset where $q = m/n$ is the sampling ratio per batch,
$\sigma=\sqrt{2\log(1.25/\delta)}/\epsilon$, and $\epsilon \leq 1$.

Using moments accounting~\cite{Abadi:2016:dpdl}, it can be proved that Algorithm~\ref{alg:basic} is $( \mathcal{O} ( q\epsilon\sqrt{t}), \delta)$-DP,
where $t$ is the total number of iterations in the main loop,
if the noise scale $\sigma$ and the clipping threshold $C$ are chosen appropriately.

\begin{theorem}
  \label{theorem:basic}
  Algorithm~\ref{alg:basic} is $( \mathcal{O} ( q\epsilon\sqrt{t}), \delta)$-DP,
  where $t$ is the total number of iterations in the main loop,
  if the noise scale $\sigma$ and the clipping threshold $C$ are chosen appropriately.
\end{theorem}

\begin{proof}

We have the following facts about moments accounting, Gaussian mechanism, and random sampling~\cite{Abadi:2016:dpdl}:
\begin{myitemize}

\item (1) Let $\mathcal{M}$ be the composition of a sequence of sub-mechanisms $\{\mathcal{M}_j\}_{j=1}^J$, it holds that $\alpha_\mathcal{M}(\lambda) \leq \sum_{j=1}^J \alpha_{\mathcal{M}_j} (\lambda)$.

\item (2) Using Markov's inequality, we have for any $\epsilon > 0$, $\mathcal{M}$
  satisfies $(\epsilon, \delta)$-DP for
  $\delta = \min_{\lambda} ( \alpha_\mathcal{M} - \lambda \epsilon )$.

\item (3) Consider a function $f$ which maps a data sample to a real-valued vector, with its output bounded by $||f||_2\leq 1$. Let $\sigma \geq 1$ and  $\mathcal{I}$ be a set of samples from $[n]$ where each $i \in \mathcal{I}$ is selected from $[n]$ independently with probability $q \leq \frac{1}{16\sigma}$. Then for any positive integer $\lambda \leq -\sigma^2 \ln (q\sigma)$, the mechanism $\mathcal{M}(d) = \sum_{i \in \mathcal{I}} f(d_i) + \mathcal{N}(0, \sigma^2\mathbf{I})$ satisfies
\begin{displaymath}
\alpha_\mathcal{M}(\lambda) \leq \frac{q^2 \lambda (\lambda +1)}{(1-q)\sigma^2} + \mathcal{O}(q^3 \lambda^3/\sigma^3)
\end{displaymath}

\end{myitemize}

Assume that $\sigma$ and $\lambda$ satisfy the condition in (3). The log-moment of Algorithm\,\ref{alg:basic} is bounded by $\alpha(\lambda) \leq q^2\lambda^2 t/\sigma^2$, according to (2) and (3). To ensure that Algorithm\,\ref{alg:basic} satisfies $(\bar{\epsilon}, \bar{\delta})$-DP, it suffices to have (i) $q^2\lambda^2 t/\sigma^2 \leq \lambda \bar{\epsilon}/2$, (ii) $\exp(-\lambda \bar{\epsilon}/2) \leq \bar{\delta}^2$, and (iii) $\lambda \leq -\sigma^2 \log(q\sigma)$.

With easy calculation, it can be verified that there exist two constants $c_1$ and $c_2$, such that when $\bar{\epsilon} = c_1 q^2 t$ and $\sigma = c_2 q\sqrt{-\log\bar{\delta}}/\bar{\epsilon}$, all the aforementioned conditions are met.
\end{proof}

\section{Optimizations}
\label{sec:opt}

The GAN formulation is known for its training stability issue~\cite{Gulrajani:2017:wganip}. This issue is even more evident in the \system framework, as random noise is injected in each training step. In our empirical study (Section~\ref{sec:eval}), it is observed that the basic \system suffers a set of drawbacks.
\begin{myitemize}
	\item Its synthesized data is often of low quality, e.g., unrealistic looking images.
	\item It converges slower than its regular GAN counterpart, resulting in excessive privacy loss, and sometimes even diverges.
	\item Its framework is fairly rigid, unable to take advantage of extra resources, e.g., a small amount of public data.

	%
	%
	%
\end{myitemize}

 Here we propose a suite of optimization strategies that significantly improve \system's training stability and convergence rate. Specifically, we enhance the basic \system along three directions.
\begin{myitemize}
\item Parameter grouping - By carefully grouping the parameters and perform stratified clipping over different groups, we strike a balance between convergence rate and privacy cost.
\item Adaptive clipping - By monitoring the change of gradient magnitudes, we dynamically adjust the clipping bounds to achieve faster convergence and stronger privacy.
\item Warm starting - By initializing the model with a good starting point, we boost up the convergence and save the privacy budget for critical iterations.
\end{myitemize}
Next we detail each of these optimization strategies.


\subsection{Parameter Grouping}

As shown in Algorithm~\ref{alg:basic}, the DP constraint essentially influences the training in two key operations (line 8): clipping - the norm of gradients is truncated by an upper bound $C$, and perturbation - random noise is added to the gradients. We propose to explore the opportunities to optimize these two critical operations.

In Algorithm~\ref{alg:basic}, the gradients of all the parameters are grouped together to compute the norm. This global clipping scheme minimizes the privacy budget spent in each iteration, but introduces excessive random noise for some parameters, causing slow convergence. At the other end of the spectrum, one may clip the gradient of each parameter with a parameter-specific clipping bound, which may reduce the overall amount of random noise, but at the cost of privacy budget. Here we propose two alternative grouping strategies that strike a balance between convergence rate and privacy loss per iteration.

%
%

 \subsubsection*{\bf Weight-Bias Separation} In most GAN architectures (e.g., convolutional layers and fully connected layers),  there are two types of parameters, weights and biases. For example, a fully connected layer models a linear function $f(x) = w\cdot x + b$ where $w$ and $b$ are the weight and bias parameters respectively. In our empirical study, it is observed that the magnitudes of the biases' gradients are often close to zero, while the magnitudes of the weights' gradients are much larger. Thus, our first strategy is to differentiate weight and bias parameters and to group the gradients of all the bias parameters together for the clipping operation. Given the large number of bias parameters, under the same amount of overall privacy budget, this strategy almost doubles the allowed number of iterations, with
little influence on the convergence rate (details in Section~\ref{sec:eval}).

%

 \subsubsection*{\bf Weight Clustering} While it is natural to group the bias parameters together as many of them are close to zero, the grouping of the weight parameters is much less obvious. Here
 we propose a simple yet effective strategy to stratify and cluster the weight parameters. Assuming that we have the optimal parameter-specific clipping bound $\{c(g_i)\}_{i}$ for each weight's gradient $\{g_i\}_{i}$  (we will show how to achieve this shortly), we then cluster these parameters into a predefined number of groups using a hierarchical clustering procedure, as sketched in Algorithm~\ref{alg:grouping}.

 Specifically, starting with each gradient forming its own group (line 1), we recursively find two groups $G, G'$ with the most similar clipping bounds and merge them to form a new group (line 3-4). As we use $\ell_2$ norm, the clipping bound of the newly formed group is computed as $\sqrt{c(G)^2 + c(G')^2}$.

  \begin{algorithm}
  	\KwIn {$k$ - targeted number of groups; $\{c(g_i)\}_i$ - parameter-specific gradient clipping bounds}
	\KwOut{$\mathcal{G}$ - grouping of parameters}
    %
    %
  	$\mathcal{G} \gets \{(g_i: c(g_i))\}_i$\;
  	\While{$|\mathcal{G}| > k $} {
  			$G, G' \gets \argmin_{G, G' \in \mathcal{G}}
  			\max \left(
  				\frac{c(G)}{c(G')}
  			 ,
  			 	 \frac{c(G')}{c(G)}
  			 \right)$ \;
             merge $G$ and $G'$ with clipping bound as $\sqrt{c(G)^2 + c(G')^2}$\;
   		}
   	\Return $\mathcal{G}$
  \caption{Weight-Clustering}
  \label{alg:grouping}
  \end{algorithm}

\subsection{Adaptive Clipping}
In Algorithm \ref{alg:basic}, the gradient clipping bound $C$ is a hyper-parameter
that needs careful tuning. Overly small $C$ amounts to excessive truncation of the gradients, while overly large $C$ is equivalent to overestimating the sensitivity, both resulting in slow convergence and poor utility. However, within the improved WGAN framework, it is challenging to find a near-optimal setting of $C$, due to reasons including: (i) the magnitudes of the weights and biases and their gradients vary greatly across different layers; and (ii) the magnitudes of the gradients are constantly changing during the training.

To overcome these challenges, we propose to constantly monitor the magnitudes of the gradients before and during the training, and set the clipping bounds based on the average magnitudes. Specifically, we assume that besides the private data $\mathcal{D}_{\rm pri}$ to train the model, we have access to a small amount of public data $\mathcal{D}_{\rm pub}$ which is available in many settings. During each training step, we randomly sample a batch of examples from $\mathcal{D}_{\rm pub}$, and set the clipping bound of each parameter as the average gradient norm with respect to this batch. In our empirical study (Section~\ref{sec:eval}), we find that this adaptive clipping strategy leads to much faster training convergence and higher data utility.



\subsection{Warm Starting}

It is expected that due to the random noise injected in each training step, the GAN with the DP constraint often converges slower than its vanilla counterpart, especially during its initial stage. To boost up the convergence rate, we propose to leverage the small amount of public data $\mathcal{D}_{\rm pub}$ to initialize the model. Specifically, using $\mathcal{D}_{\rm pub}$, we first train a few iterations without the DP constraint, and then continue the training using $\mathcal{D}_{\rm pri}$ under the DP constraint.

This strategy provides a warm start for \system. It helps find a satisfying starting point, which is essential for the model to converge, and also saves a significant amount of privacy budget for more critical iterations.

An astute reader may point out that since there is public data available, one may just use the public data for training. The issue is that the public data is often fairly limited, which may not be sufficient to train a high-quality GAN. Further, the large amount of private data is valuable for improving the diversity of the samples synthesized by the generator (details in Section~\ref{sec:eval}).

%
%

\SetKwProg{myproc}{Procedure}{}{}
 \begin{algorithm}
	\caption{Advanced \system}
	\label{alg:advanced}
	\KwIn {$n$ - number of samples; $\mathcal{D}_{\rm pub}$ - public dataset;
	$\lambda$ - coefficient of gradient penalty; $n_\text{critic}$ - number of critic iterations
		per generator's iteration; $n_\text{param}$ - number of discriminator's parameters; $m$ - batch size for training GAN;
		$m_\text{pub}$ - batch size for estimating norms of gradients;
		($\alpha,
		\beta_1, \beta_2$) - Adam hyper-parameters; $C$ - gradient clipping bound; $\sigma$ - noise
		scale; ($\epsilon_0$, $\delta_0$) - overall privacy target; $k$ - number of parameter groups}
	\KwOut{$G$ - differentially private generator}
	\tcp{warm starting}
	$\left(w, \theta \right) \gets $ train regular improved
	 WGAN using $\mathcal{D}_{\rm pub}$\;
	\While{$\theta$ has not converged} {
		\For {$t = 1, \cdots, n_\text{critic}$} {
		 \tcp{computing gradients of public data}
			sample $\{\bar{x}_i\}_{i=1}^{m_{\rm pub}} \sim \mathcal{D}_\text{pub}$ \;
			$\{\bar{g}^{(i)}\}^{m_\text{pub}}_{i = 1} \gets$ Improved WGAN-Gradient~
		($\{\bar{x}_i\}_{i=1}^{m_{\rm pub}}$, $m_{\rm pub}$)\;
		\tcp{grouping parameters with similar clipping bounds}
						$\{(G_j, c_j)\}_{j = 1}^k \gets$ Weight-Clustering~
			($k$, $\{\bar{g}^{(i)}\}^{m_\text{pub}}_{i = 1}$)\;
			\tcp{computing gradients of real data}
			sample $\{x_i\}_{i=1}^m \sim p_{\rm data}$ \;
			$\{g^{(i)}\}^{m}_{i = 1} \gets$ Improved WGAN-Gradient~
		($\{x_i\}_{i=1}^{m}$, $m$)\;
			\For {$i = 1, \cdots, m$} {
				${g^{(i)}_j} \gets g^{(i)} \cap G_j$ for $j = 1, \cdots, k$ \;
				\For {$j = 1, \cdots, k$}{
					\tcp{clipping and perturbation $\xi \sim \mathcal{N}(0, (\sigma c_j)^2\mathcal{I})$}
					$g^{(i)}_j \gets g^{(i)}_j / \max(1, || g^{(i)}_j ||_2/c_j ) + \xi$\;
				}
				}
			\tcp{updating privacy accountant}
			update $\mathcal{A}$ with $(\sigma, m, k)$ \;
			\tcp{updating discriminator}
			$w_j \gets \text{Adam}( \frac{1}{m} \sum_{i = 1}^m g^{(i)}_j,
			w_j, \alpha, \beta_1, \beta_2 )$ for $j = 1, \cdots, k$\;
			$w \gets \left\{w_j\right\}_{j = 1}^k$\;

		}
		sample $\{ z^{(i)} \}_{i = 1}^m \sim p_z$ \;
		\tcp{updating generator}
		$\theta \gets \text{Adam}( \nabla_\theta \frac{1}{m} \sum_{i = 1}^m
		-D( G(z^{(i)})) , \theta, \alpha, \beta_1, \beta_2
		)$\;
		\tcp{computing cumulative privacy loss}
		$\delta \gets$ query $\mathcal{A}$ with $\epsilon_0$\;
		\lIf {$\delta \geq \delta_0$} {
			break
		}
		}
		\Return{G}\;
		\myproc{\rm Improved WGAN-Gradient~($\{x_i\}_{i=1}^m, m$)}
		{
			\For {$i = 1, \cdots, m $}{
				sample $z \sim p_z$, $\rho \sim \mathcal{U}[0, 1]$\;
				$\hat{x} \gets \rho x_i + (1 - \rho)  G(z) $ \;
				$\ell^{(i)} \gets D(G (z)) - D(x_i) +
				\lambda ( || \triangledown_{\hat{x}}  D( \hat{x})
				||_2 - 1 )^2$ \;
				$g^{(i)} \leftarrow \triangledown_{w} \ell^{(i)}$\;
			}
			\Return $\{ g^{(i)} \}_{i = 1}^m$\;
		}
\end{algorithm}

\subsection{Advanced Algorithm}

%


Putting everything together, Algorithm~\ref{alg:advanced} sketches the
enhanced \system framework. Different from Algorithm~\ref{alg:basic}, we initialize the
model with a warm starting procedure using the public data $\mathcal{D}_\text{pub}$ (line 1). During each training iteration, we first estimate the clipping bound of each parameter using $\mathcal{D}_\text{pub}$ (line 4-5), then group the parameters into $k$ groups $\{G_j\}_{j=1}^k$, each $G_j$ sharing similar clipping bound $c_j$ (line 6). In our current implementation, we use the average clipping bounds in $G_j$ to estimate $c_j$. We then perform group-wise clipping and perturbation (line 9-12). The remaining part is similar to Algorithm~\ref{alg:basic}. The process iterates until the generator's parameters converge or the privacy budget is used up (line 19).

Astute readers may raise the concern about possible additional privacy loss due to the multiple optimization strategies. We have the following theorem.

\begin{theorem}
  Algorithm~\ref{alg:advanced} is $( \mathcal{O} ( q\epsilon\sqrt{t}), \delta)$-DP,
  where $t$ is the total number of iterations in the main loop,
  if the noise scale $\sigma$ and the clipping threshold $C$ are chosen appropriately.
\end{theorem}

\begin{proof}
 Algorithm~\ref{alg:advanced} differs from Algorithm~\ref{alg:basic} mainly in its use of finer-grained clippings for different groups of parameters, which however does not cause additional privacy loss. Intuitively, thanks to the composability property of the moments accounting~\cite{Abadi:2016:dpdl}, the privacy loss due to applying parameter-specific clipping is completely accounted.

Next we prove that the strategy of weight clustering does not cause unaccounted privacy loss, while similar arguments apply to other optimization strategies as well.

In Algorithm\,\ref{alg:basic}, in the $i$-th iteration, the gradient $g^{(i)}$ is first clipped by a global bound $c$ and the random noise $\xi \sim \mathcal{N}(0, (\sigma c)^2\mathbf{I})$ is applied to $g^{(i)}$ to ensure $(\epsilon, \delta)$-DP, where $\sigma = \sqrt{2 \log(1.25/\delta)}/\epsilon$.

In Algorithm~\ref{alg:advanced}, $g^{(i)}$ is divide into $k$ sub-vectors $\{g_j^{(i)}\}_{j =1}^k$. Each sub-vector $g_j^{(i)}$ is clipped by a group-specific bound $c_j$ and the random noise $\xi_j \sim \mathcal{N}(0, (\sigma c_j)^2\mathbf{I})$ is applied, where $\sigma = \sqrt{2 \log(1.25/\delta)}/\epsilon$. Thus, releasing each $g_j^{(i)}$ satisfies $(\epsilon, \delta)$-DP. As $\{g_j^{(i)}\}_{j =1}^k$ are disjoint, applying the parallel composability property of DP~\cite{Dwork:2014:book}, releasing $\{g_j^{(i)}\}_{j =1}^k$ also satisfies $(\epsilon, \delta)$-DP.

\end{proof}

\section{Empirical Evaluation}
\label{sec:eval}

In this section, we empirically evaluate the proposed \system framework. The experiments are designed to answer four key questions that impact \system's practical use. First, is \system able to synthesize visually vivid image data, under the DP constraint? Second, does the synthesized data demonstrate sufficient quality and diversity, from a quantitative perspective? Third, does the synthesized data retain enough utility for concrete data analysis tasks? Finally, how do different optimization strategies influence \system's performance?
We begin with describing the experimental setting.

%

\subsection{Experimental Setting}

In our experiments, we use three benchmark datasets:
\begin{myitemize}
\item MNIST, which consists of 70K handwritten digit images of size $28\times28$, split into 60K training and 10K test samples.
\item CelebA, which comprises 200K celebrity face images of size $48\times 48$, each with 40 attribute annotations.
\item LSUN, which contains around one million labeled images of size $64\times 64$, for each of the 10 scene categories.
\end{myitemize}


For the MNIST and CelebA datasets, we split the training data (which is the entire dataset if no labeling information is considered) using the ratio of
$2:98$ as publicly available data $\mathcal{D}_{\rm pub}$ and private data $\mathcal{D}_{\rm pri}$ respectively. We train \system on
$\mathcal{D}_{\rm pri}$ under the DP constraint.
For the LSUN dataset, we consider two settings. First, we consider it as an unlabeled dataset and split it into $2:98$ as public data $\mathcal{D}_{\rm pub}$ and private data $\mathcal{D}_{\rm pri}$, which we denote as
LSUN-U. Second, we consider the label information of the dataset. We sample 500K images from each of the top 5 categories (in terms of number of images), which are then split into $2 : 98$ as $\mathcal{D}_{\rm pub}$ and $\mathcal{D}_{\rm pri}$ respectively. We refer to this dataset as LSUN-L.

The network architecture of \system is similar to~\cite{Gulrajani:2017:wganip}, which we adpat to each dataset.
The default setting of the parameters is as follows: the coefficient of gradient penalty $\lambda=10$,
the number of critic iterations per GAN's iteration $n_\text{critic} = 4$, the batch size $m = 64$.
The setting of the parameters specific to each dataset is summarized in Table~\ref{tab:expsetting},
where $(\alpha, \beta_1, \beta_2)$ are the hyper-parameters of the Adam optimizer, $(\epsilon, \delta)$ are the privacy budget, and $\sigma$ is the noise scale. The setting of $\sigma$ follows the setting in~\cite{Abadi:2016:dpdl}, which is considered sufficiently strict in typical applications.
 The last two hyper-parameters are for advanced \system: $k$ is the number of groups for weight clustering, and $t_\text{warm}$ is the number of iterations for warm starting with public data.


\begin{table}
	\centering
	\begin{tabular}{c | c c c c c c c c}
		Dataset & $\alpha$ & $\beta_1$ & $\beta_2$ & $\epsilon$ & $\delta$ & $
		\sigma$ & $k$ & $t_\text{warm}$ \\
		\hline
		MNIST & 0.002 & 0.5 & 0.9 & 4 & $10^{-5}$ & 1.086 & 5 & 300 \\
		CelebA & 0.002 & 0.0 & 0.9 & 10 & $10^{-5}$ & 0.543  & 6 & 800  \\
		LSUN-U  & 0.002 & 0.0 & 0.9 & 10 & $10^{-5}$ & 0.434 & 7 & 2400 \\
		LSUN-L  & 0.002 & 0.0 & 0.9 & 10 & $10^{-5}$ & 0.434 & 7 & 2000 \\
		\hline
	\end{tabular}
	\caption{Parameter setting for each dataset.}
	\label{tab:expsetting}
\end{table}



All the experiments are conducted on TensorFlow.

\subsection{Qualitative Evaluation}

 In this set of experiments, we qualitative evaluate the quality of the data synthesized by \system. Figure~\ref{fig:mnist},~\ref{fig:lsunbedroom},~\ref{fig:lsun10cat}, and ~\ref{fig:celeba48} show a set of synthetic samples generated by \system, which has been trained on the MNIST, LSUN-U, LSUN-L, and CelebA datasets respectively. It is noted that in all the cases, \system is able to generate visually vivid images of quality comparable to original ones, while, at the same time, providing strong privacy protection (see Table~\ref{tab:expsetting}).
\begin{figure*}[t]
	\centering
	\epsfig{width = 175mm, file = 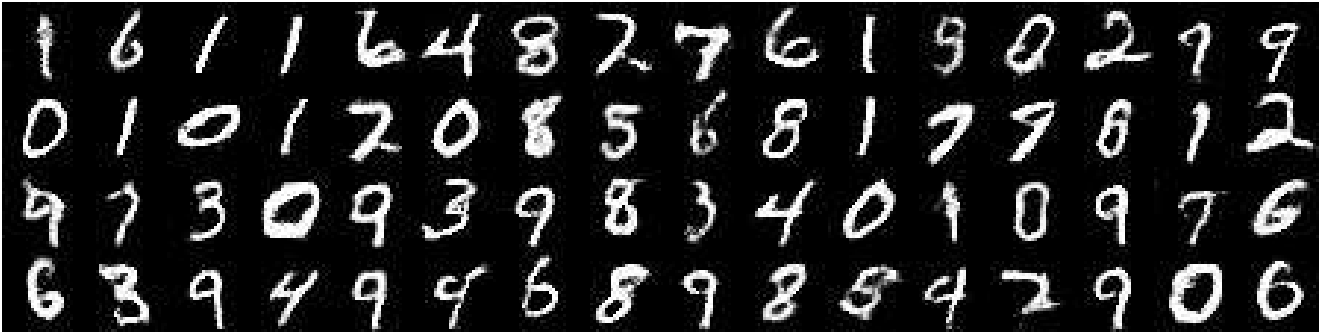}
	\caption{Synthetic samples for the MNIST dataset ($\epsilon=4, \delta \leq 10^{-5}$)}
	\label{fig:mnist}
\end{figure*}

\begin{figure*}[t]
	\centering
	\epsfig{width = 175mm, file =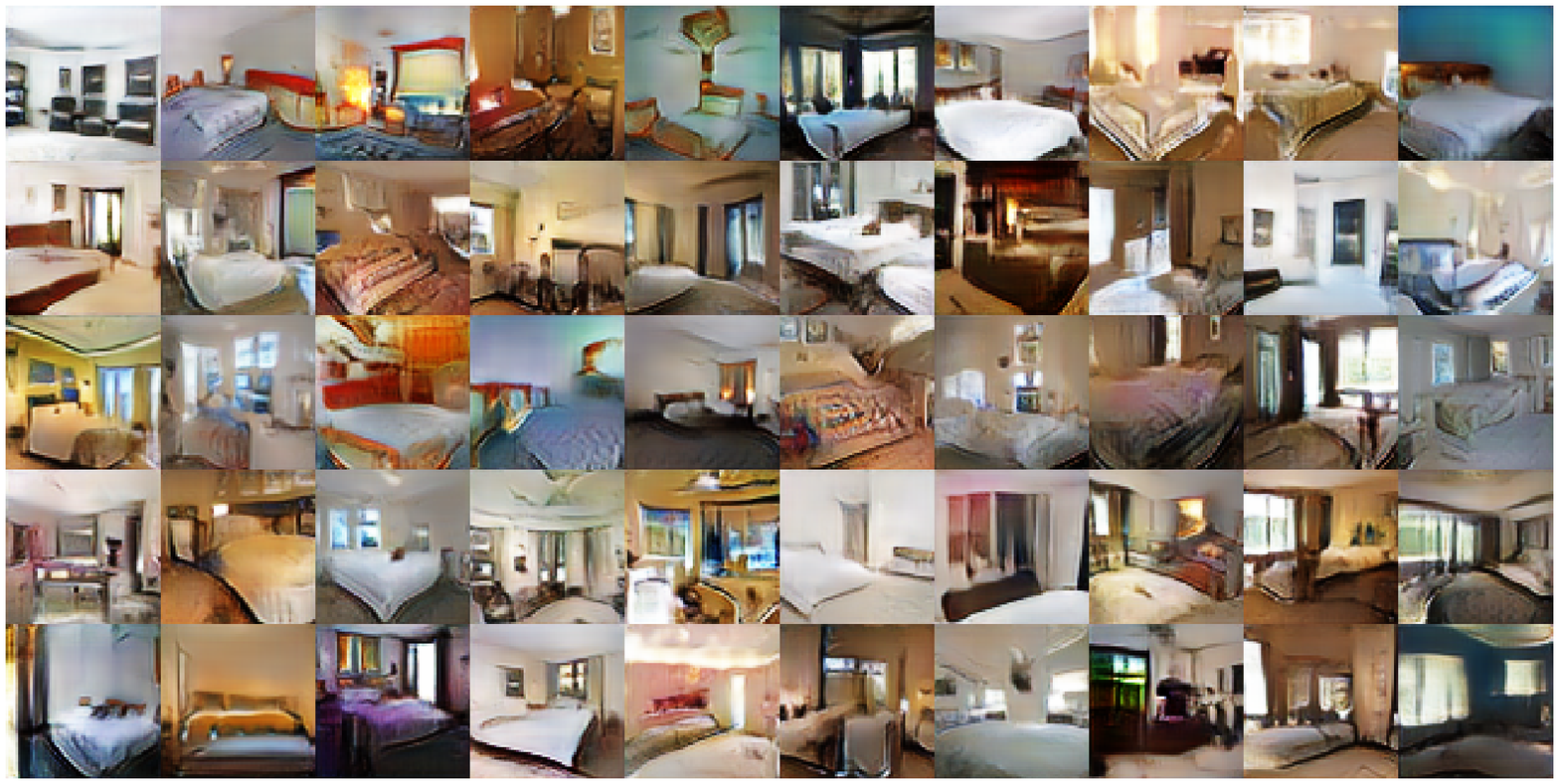}
	\caption{Synthetic samples for LSUN-U dataset ($\epsilon=10, \delta \leq 10^{-5}$)}
	\label{fig:lsunbedroom}
\end{figure*}

\begin{figure*}[t]
	\centering
	\epsfig{width = 175mm, file = 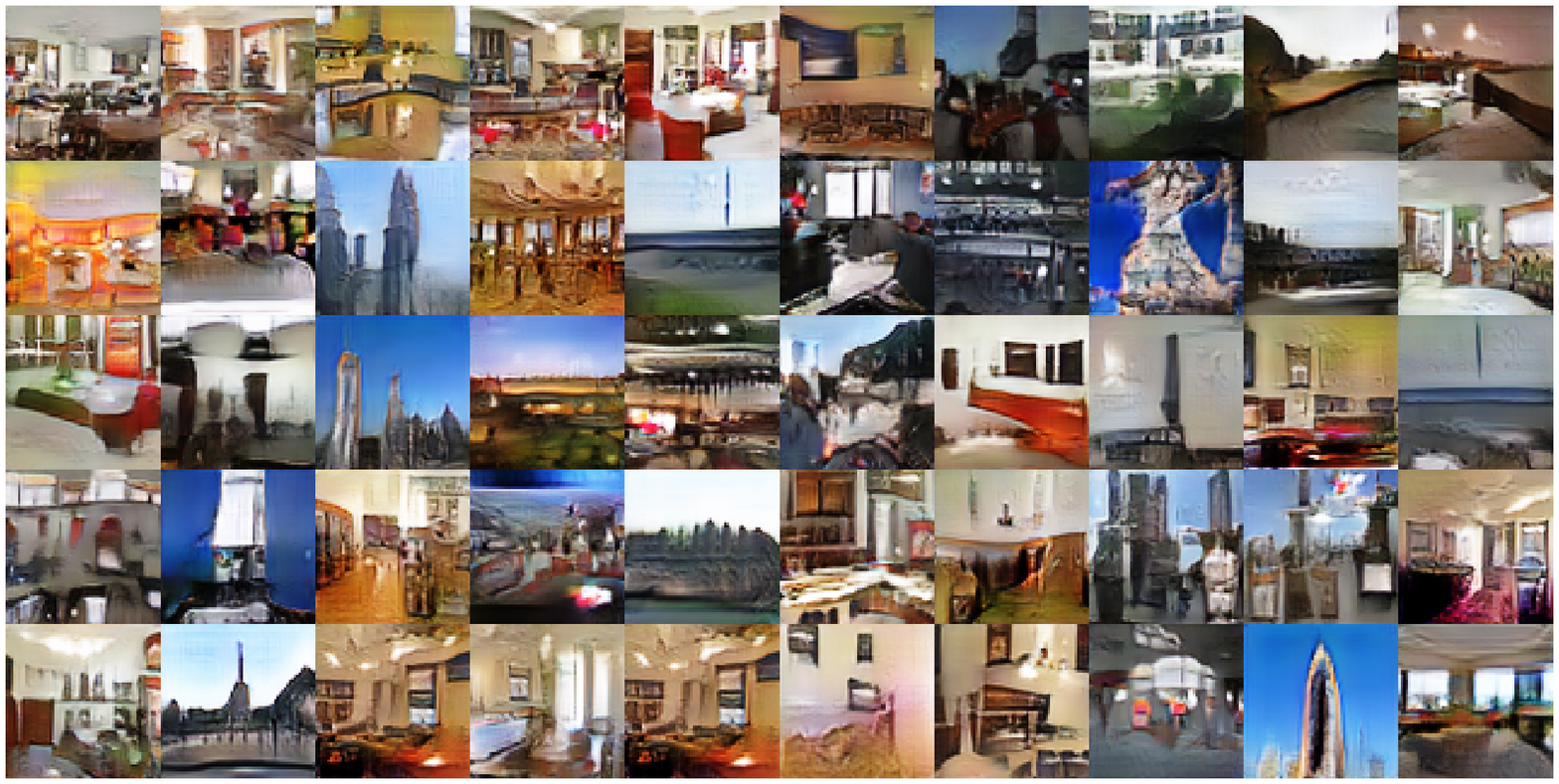}
	\caption{Synthetic samples for the LSUN-L dataset ($\epsilon=10, \delta \leq 10^{-5}$)}
	\label{fig:lsun10cat}
\end{figure*}

\begin{figure*}[t]
	\centering
	\epsfig{width = 175mm, file = 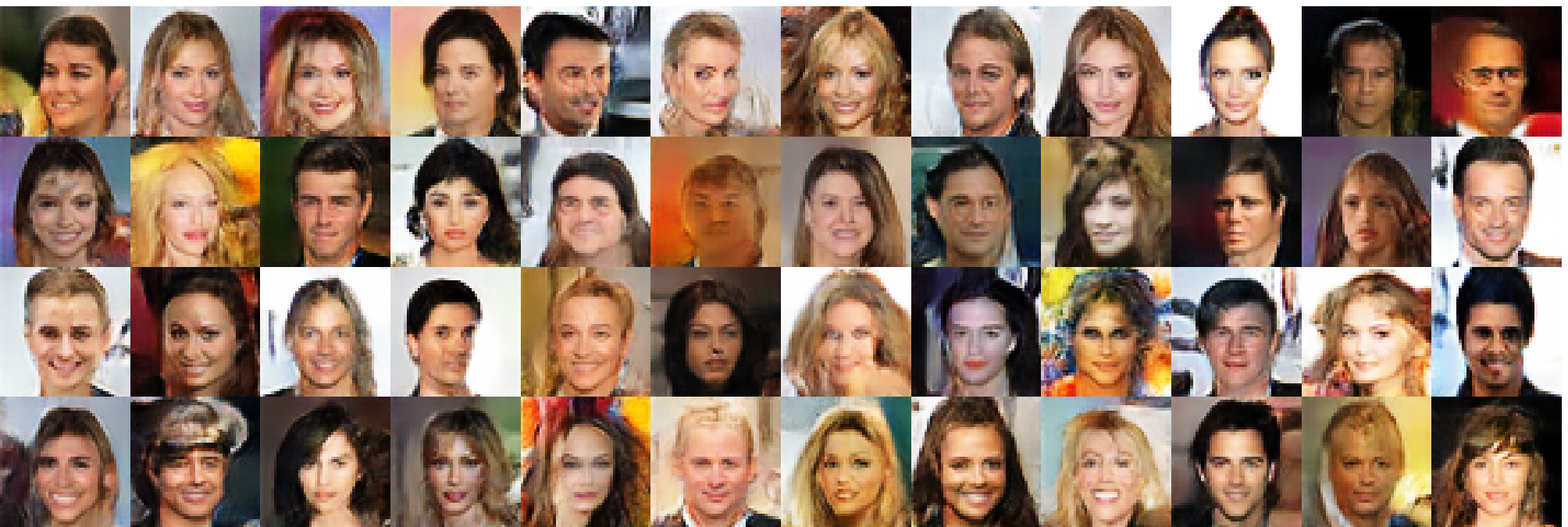}
	\caption{Synthetic samples for the CelebA dataset ($\epsilon=10, \delta \leq 10^{-5}$)}
	\label{fig:celeba48}
\end{figure*}

\subsection{Quantitative Evaluation}

Next we conduct quantitative evaluation of \system's performance. Specifically, we first compare the synthetic data against the real data in terms of their statistical properties, including Inception scores and Jensen-Shannon divergence; we then evaluate the quality of the synthetic data in semi-supervised classification tasks.

\subsubsection*{\bf Statistical Properties}

In \cite{Salimans:2016:improved}, Salimans {\em et al.} propose to use Inception score to measure the quality
of data generated by GAN. Formally, the Inception score\footnote{Even though the datasets here are not ImageNet, we still refer to
	Eqn.~\ref{equ:is} as Inception score in the following.} of a generator $G$ is defined as:
\begin{equation}
s(G) = \exp \left( \mathbb{E}_{x \sim G(z) }
{\rm KL}( {\rm Pr}(y|x) ||   {\rm Pr}(y)) \right)
\label{equ:is}
\end{equation}
Here, (i) $x$ is a sample generated by $G$.
(ii) ${\rm Pr}(y|x)$ is the conditional distribution imposed
by a pre-trained classifier
to predict $x$'s label $y$. If $x$ is similar to a real sample, we expect the entropy of ${\rm Pr}(y|x)$ to be small. (iii)
${\rm Pr}(y) = \int_{x} {\rm Pr}(y|x = G(z)
{\rm d}z$ is the marginal distribution of $y$. If $G$ is able to generate a diverse set of samples, we expect the entropy of ${\rm Pr}(y)$ to be large. Thus, by measuring the KL divergence of the two distributions, $s(G)$ captures both the quality and diversity of the synthetic data. For the MNIST and LSUN-L datasets, we use the entire training set to train baseline classifiers to estimate ${\rm Pr}(y|x)$. The classifiers are tuned to achieve reasonable performance on the validation sets (99.06\% for MNIST and 88.73\% for LSUN-L).


Table~\ref{tab:iscore} summarizes the Inception scores of synthetic data (generated by regular GAN and \system) and real data for the MNIST and LSUN-L datasets. It can be noticed that \system is able to synthesize data with Inception scores fairly close to the real data and that generated by regular GANs (without privacy constraints). For example, in the case of MNIST, the difference between the real data and the synthetic data by \system is less than 1.32.

\begin{table}
	\centering
	\begin{tabular}{ c | c c c  c  }
		Dataset & Setting & $n$ ($\times 10^6$)  &  $(\epsilon, \delta)$ & Score \\
		\hline
		\multirow{3}{*}{MNIST} & real & $0.06$ &  - & $9.96 \pm 0.03 $\\
		& GAN  & $0.06$ & -  & $9.05 \pm 0.03$ \\
		& \system & $0.05$ & $(4, 10^{-5})$ & $8.64 \pm 0.03 $\\
		\hline
		\hline
		\multirow{3}{*}{LSUN-L} & real & $2.50$ & - &  $4.16 \pm 0.01 $\\
		& GAN & $2.50$ &  - &  $3.11 \pm 0.01$ \\
		& \system & $2.45$ & $(10, 10^{-5}) $ &  $2.78 \pm 0.01$ \\
		\hline
	\end{tabular}
	\caption{Inception scores of real and synthetic data on the MNIST and LSUN-L datasets (with label information).}
	\label{tab:iscore}
\end{table}
\begin{table}
	\centering
	\begin{tabular}{c | c c c c}
		Dataset & Setting & $n$ ($\times 10^6$) & $(\epsilon,\delta)$ & Score \\
		\hline
		\multirow{3}{*}{CelebA} & real & $0.22 $ & - & $0.00 \pm 0.00$ \\
		& GAN &  $0.20 $ & - & $0.09 \pm 0.00$\\
		& \system &  $0.20 $ & $(10, 10^{-5}) $ & $0.28 \pm 0.00$ \\
		\hline
		\hline
		\multirow{3}{*}{LSUN-U} & real & $2.50 $ & - & $0.00 \pm 0.00$ \\
		& GAN & $2.50  $ & - & $0.25 \pm 0.00 $\\
		& \system & $ 2.45 $ & $(10, 10^{-5})$ & $0.29 \pm 0.00 $\\
		\hline
	\end{tabular}
	\caption{Jensen-Shannon scores of real and synthetic data on CelebA and LSUN-U datasets (without label information).}
	\label{tab:iscorewolab}
\end{table}

To measure \system's performance with respect to unlabeled data (e.g., CelebA and LSUN-U), we train another discriminator $D'$ using the real data and test whether $D'$ is able to discriminate the synthetic data. We consider two distributions: (i)
${\rm Pr}(y|x)$ is the conditional distribution that $D'$'s prediction about $x$'s source (real or synthetic) and (ii) $\mathcal{B}_p$ is a Bernoulli distribution with $p=0.5$. We use the
Jensen-Shannon divergence of the two distributions to measure the quality of the synthetic data:
\begin{displaymath}
	s(G) = \frac{1}{2} \text{KL} ( {\rm Pr}( y | x) || \mathcal{B}_p ) +
	\frac{1}{2} \text{KL} ( \mathcal{B}_p || {\rm Pr}( y | x ) )
\end{displaymath}
Intuitively, a smaller value of $s(G)$ indicates that $D'$ has more difficulty to discriminate the synthetic data, i.e., better quality of the data generated by $G$.

Table~\ref{tab:iscorewolab} summarizes the quality scores of the real and  synthetic data (regular GAN and \system) on the
CelebA and LSUN-U datasets. Observe that \system generates data of quality close to that by regular GAN (without privacy constraints), especially in the case of LSUN-U, i.e., 0.25 versus 0.29. This may be explained by that compared with CelebA, LSUN-U is a relatively larger dataset, enabling \system to better capture the underlying data distribution.

GANs that generate images with a single label or without a explicit classification task.
We can consider train another discriminator $D$ to produce a signal to identify if
the images come from generative distribution $G\left(p \left(z \right)\right)$~$
(z \sim p(z))$ or real data distribution $p_\text{data}$. We donate $y = 1$ if an image $x$
comes from real distribution, and $y = - 1$ if it is a generated one. Here we use
\textit{Jensen--Shannon divergence} to measure the distance between two distributions.
	Specifically, we take
\begin{equation}
	s(G) = \frac{1}{2} \text{KL} \left( p\left( y | x\right)  \middle\Vert q\left(y\right) \right) +
	 	\frac{1}{2} \text{KL} \left( q\left(y\right) \middle\Vert p\left( y | x \right) \right)
\end{equation}
to measure the quality of generated examples without an explicit supervised task.
Here $p\left(y | x\right)$ is the conditional
probability of a sample $x$ is coming from $p_\text{data}$, which is a Bernoulli distribution,
and $q(y)$ is a Bernoulli distribution with parameter $p = 0.5$. Thus,
the better the generator, the smaller its score, as even a good discriminator is not able to
decide if the sample $x$ is coming from the generator $G$ or $p_\text{data}$.
In table~\ref{tab:iscorewolab}, we show the quality of scores of unlabeled images with
CelebA and LSUN--Bedroom datasets.
\begin{table}
	\centering
	\begin{tabular}{c | l c c c}
		Dataset & Setting & $n$ ($\times 10^6$) & $(\epsilon,\delta)$ & Score \\
		\hline
		\multirow{3}{*}{CelebA} & real & $0.22 $ & - & $0.00 \pm 0.00$ \\
		& sync.~(w/o DP) &  $0.20 $ & - & $0.09 \pm 0.00$\\
		& sync.~(w/ DP) &  $0.20 $ & $(10, 10^{-5}) $ & $0.28 \pm 0.00$ \\
		\hline
    \hline
		\multirow{3}{*}{LSUN-U} & real & $2.50 $ & - & $0.00 \pm 0.00$ \\
		& sync.~(w/o DP) & $2.50  $ & - & $0.25 \pm 0.00 $\\
		& sync.~(w/ DP) & $ 2.45 $ & $(10, 10^{-5})$ & $0.29 \pm 0.00 $\\
		\hline
	\end{tabular}
	\caption{Inception scores for generated examples and real examples ON
	CelebA and LSUN--Bedroom.}
	\label{tab:iscorewolab}
\end{table}

\subsubsection*{\bf Analysis Tasks}

We further evaluate \system's performance in concrete analysis tasks. Specifically, we consider the use of synthetic data in a semi-supervised classification task. In such a task, the analyst possesses a small amount of public, labeled data and a large amount of synthetic, unlabeled data (generated by \system). The goal is to leverage both the labeled and unlabeled data to train a better classifier than that trained only using the limited labeled data.

%

To make things more interesting,
we consider the setting of two separate classifiers. The first one $\mathcal{C}_1
$ has the same structure as a regular image
classifier; while the second one $\mathcal{C}_2$
classifies using both an image and its code. The
architecture of $\mathcal{C}_2$ is designed to learn the correlation between the codes and the
images. The learning procedure is sketched in Algorithm~\ref{alg:semi}, it consists of two part for
each iteration.
In the first part (line 2-5), we first sample a batch of $m$ codes $\hat{z}$,
and generate images $\hat{z}$ from generator $G$ with $\hat{z}$, and use $\mathcal{C}_1$ to
classify $\hat z$ into category $\hat y$. Then we update $\mathcal{C}_2$ with $(\hat{z}, \hat{x}, \hat{y})$~
(line 5). In the second part (line 6-9), we sample a batch of $m\cdot(1 - p_s)$ real examples $(x, y)$ from the labeled data, and
then sample another batch of $m \cdot p_s$ codes $\hat{z}$ and their synthetic images $ \hat{z}$, and labeled them with $\mathcal{C}_2$ as $\hat{y}$.
Now we take both sets of inputs to update $\mathcal{C}_1$.

We hope that $\mathcal C_1$ and $\mathcal C_2$ would converge quickly.
However, In the experiments, we found that if we use the data from $\mathcal C_2$ too early, it would cause
the entire model unstable, and difficult to converge to a proper accuracy, due to that $\mathcal{C}_2$ is not
fully trained with correct labels (i.e., in early state, both $\mathcal C_1$ and $\mathcal C_2$ have low accuracy).
 Thus,  In practice, it is sensible to increase $p_s$ gradually during the training after some iterations. In our experiments, we first introduce
 $p_s = 0$ at the one third point of the regular model, and gradually increase it to $p_\text{s, final}$. Then we follow the flow in Algorithm~\ref{alg:semi} with
 $p_s = p_\text{s, final}$.

%
%

\begin{algorithm}
	\caption{Semi-Supervised Classification}
	\label{alg:semi}
	\KwIn{$m$ - batch size; $p_s$ - percentage of synthetic data in training; $G_\theta$ - privacy-preserving generator; $\mathcal{D}_{\rm pub}$ - public labeled dataset}
	\KwOut{$\mathcal C_1^{\theta_1}$ - image classifier; $\mathcal{C}_2^{\theta_2}$ - image \& code classifier
	}
	\While{$\mathcal C_1^{\theta_1}$ or  $\mathcal{C}_2^{\theta_2}$ not converged yet}{
		\tcp{training $\mathcal{C}_2$}
		sample $\{\bm \hat z_i\}_{i = 1}^m \sim p_z$ \;
		generate $ \{ \hat{x}_i\}_{i = 1}^m$ with $G_\theta$ and $\{\hat z_i \}_{i = 1}^m$ \;
		$\{\hat y_i\}_{i = 1}^m \gets$ $\mathcal C_1^{\theta_1}$~($\{ \hat{x}_i\}_{i = 1}^m$) \;
		update $\mathcal C_2$ with $(\theta_2, \{ (\hat{z}_i,  \hat{x}_i,  \hat{y}_i) \}_{i = 1}^m )$ \;
		\tcp{training $\mathcal{C}_1$}
		sample $\{\hat z_i\}_{i = 1}^{m \cdot p_s} \sim p_z$ \;
		generate $ \{\hat x_i\}_{i = 1}^{m \cdot p_s}$ with $G_\theta$ and $\{ \hat z_i \}_{i = 1}^{m \cdot p_s}$ \;
		$\{\hat y_i\}_{i = 1}^{m \cdot p_s} \gets$ $\mathcal C_2^{\theta_1}$~($ \{(\hat{z}_i, \hat{x}_i)\}_{i = 1}^{m \cdot p_s}$) \;
		sample $\{ (x_i, y_i) \}_{i = 1}^{m \cdot (1 - p_s)}$ from $\mathcal{D}_{\rm pub}$ \;
		update $\mathcal C_1$ with $\left(\theta_1,
		\{ (\hat{x}_i, \hat{y}_i) \}_{i = 1}^{m \cdot p_s},  \{  (x_i, y_i) \}_{i = 1}^{m \cdot (1 - p_s)}\right)$\;
	}
	\Return $\mathcal C_1^{\theta_1}$, $\mathcal{C}_2^{\theta_2}$
	\label{alg:semisupervised}

\end{algorithm}

We evaluate \system's performance in such a task on the LSUN-L dataset.
 It is clear that the semi-supervised classifier steadily outperforms the supervised classifier. The difference is especially evident when the size of the public data is small (i.e., limited number of labeled samples). For example, for $n = 0.5\times 10^4$, the semi-supervised classifier outperforms the supervised one by more than 6\%. We can thus conclude that \system supplies valuable synthetic data for such semi-supervised classification tasks.

%
%
\begin{table}
	\centering
	\caption{Semi-supervised Classification Task Result (LSUN-L)}
	\label{tab:semi}
	\begin{tabular}{c | c c  c c }
		Setting & n ($\times 10^4$) & $p_\text{s, final} $ & Original accuracy & Semi accuracy\\
		\hline
		\multirow{4}{*}{GAN} & 0.5 & 0.2 & 0.538 & 0.615\\
		& 1.5 & 0.2 & 0.650 & 0.661\\
		& 2.5 & 0.2 & 0.665& 0.699\\
		& 5.0 & 0.2 & 0.733 & 0.755 \\
		\hline
		\multirow{4}{*}{\system} & 0.5 & 0.2 & 0.538 & 0.571 \\
		& 1.5 & 0.2 & 0.650 & 0.669 \\
		& 2.5 & 0.2 & 0.665 & 0.695 \\
		& 5.0 & 0.2 & 0.733 & 0.737
	\end{tabular}
\end{table}

\begin{table*}
	\centering
	\caption{Semi--supervised classification tasks}
	\label{tab:semi}
	\begin{tabular}{c | c c  c c }
		Dataset & \#N~(public) & Fraction & Original accuracy & Semi accuracy\\
		& \\
		\hline
		\multirow{5}{*}{MNIST~(No DP)} & 100 & 0.2 & 0.720 & 0.567 \\
		& 200  & 0.2 & 0.824 & 0.802\\
		& 300  & 0.2 & 0.874 & 0.824 \\
		& 500  & 0.2 & 0.909  & 0.895 \\
		& 1000 & 0.2 & 0.934 & 0.923\\
		\hline

		\multirow{5}{*}{MNIST~(DP)} & 100 & 0.2 & 0.720 & 0.569 \\
		& 200  & 0.2 & 0.824 & 0.747 \\
		& 300  & 0.2 & 0.874 & 0.835 \\
		& 500  & 0.2 & 0.909 &  0.884 \\
		& 1000 & 0.2 & 0.934 & 0.903  \\
		\hline

		\multirow{4}{*}{LSUN--5 Cat~(No DP)} & 5000 & 0.2 & 0.538 & 0.615\\
		& 15000 & 0.2 & 0.650 & 0.661\\
		& 25000 & 0.2 & 0.665& 0.699\\
		& 50000 & 0.2 & 0.733 & 0.755 \\
		\hline
		\multirow{4}{*}{LSUN--5 Cat~(DP)} & 5000 & 0.2 & 0.538 & 0.571 \\
		& 15000 & 0.2 & 0.650 & 0.669 \\
		& 25000 & 0.2 & 0.665 & 0.695 \\
		& 50000 & 0.2 & 0.733 & 0.737
	\end{tabular}
\end{table*}

\subsection{Effectiveness of Optimizations}

In the final set of experiments, we evaluate the impact of different optimization strategies on \system's performance.

We first measure the strategy of weight clustering on the number of allowed iterations given the same privacy constraints. Table~\ref{tab:iternums} compares the number of allowed iterations before and after applying the weight clustering strategy. It is clear that across all the datasets, this strategy significantly increases the number of allowed iterations, thereby improving the retained utility in the generative models.

We further measure the impact of different configures of multiple optimization strategies on \system's performance with results listed in Table~\ref{tab:optimquality} and Table~\ref{tab:optimqualityun} for the labeled and unlabeled datasets respectively. It is observed that in general, combining multi-fold optimizations significantly boosts \system's performance. For example, in the case of Inception score, the score is increased from 6.59 to 8.64.

\begin{table}
	\centering
	\caption{Effect of weight grouping on the number of allowed iterations: $t_\text{before}$ and $t_\text{after}$ are respectively the maximum number of iterations
	under privacy constraint.}
	\label{tab:iternums}
	\begin{tabular}{c | c c }
		Dataset & $t_\text{before}$ & $t_\text{after}$ \\
		\hline
		MNIST & 560 & 780 \\
		CelebA & 2910 & 4070 \\
		LSUN-L & 15010 & 19300 \\
		LSUN-U & 24630 & 31670 \\
		\hline
	\end{tabular}

\end{table}

\begin{table}
	\centering
	\caption{Impact of optimization on Inception scores (MNIST) $S_1$ - Weight/Bias Separation, $S_2$ - Automated Weight Grouping, $S_3$ - Adaptive Clipping, $S_4$ - Warm Starting}
	\label{tab:optimquality}
	\begin{tabular} {c | c  c  c c c c}
		Strategy & \multicolumn{6}{c}{Configuration} \\
		\hline
		\hline
		$S_1$ &  & \checkmark & & \checkmark & & \checkmark \\
		$S_2$ & & & & & \checkmark & \\
		$S_3$ &  & & \checkmark & \checkmark & \checkmark & \checkmark \\
		$S_4$ & & & & & & \checkmark \\
		\hline
		\hline
		\multirow{2}{*}{Score} & $6.59$  & $6.46$ & $7.76$
		&  $8.20$ & $8.03$ & \bm{$8.64$} \\
		&    $\pm  0.03 $ & $\pm 0.04 $ & $\pm 0.05 $
		&  $\pm 0.02$ & $\pm 0.04 $ & \bm{$\pm 0.03$}\\
		\hline
	\end{tabular}
\end{table}

\begin{table*}
	\centering
	\caption{Effectiveness of optimizations with respect to quality
	scores~(labeled)}
	\label{tab:optimquality}
	\begin{tabular} {c | c  c  c c c }
		Dataset & Weight/Bias Grouping & Automatic Grouping & Gradient estimation &
		Pre--trained & Score  \\
		\hline
		\multirow{6}{*}{MNIST}    & & & &  & $6.59 \pm  0.03 $\\  
		 & \checkmark & & & & $6.46 \pm 0.04 $ \\ 
		 & & & \checkmark & & $7.76 \pm 0.05 $ \\ 
		 & \checkmark & &  \checkmark & &  $8.20 \pm 0.02$\\
		 & & \checkmark & \checkmark & & $8.03 \pm 0.04 $ \\
		 & \checkmark &  & \checkmark & \checkmark & \bm{$8.64 \pm 0.03$} \\
		\hline
	\end{tabular}
\end{table*}

\begin{table}
	\centering
	\caption{Effectiveness of optimizations with respect to quality
		scores~(unlabeled)}
	\label{tab:optimqualityun}
	\begin{tabular} {c | c  c  c c c c}
		Strategy & \multicolumn{6}{c}{Configuration} \\
		\hline
		\hline
		$S_1$ &  & \checkmark & & \checkmark & & \checkmark \\
		$S_2$ & & & & & \checkmark & \\
		$S_3$ &  & & \checkmark & \checkmark & \checkmark & \checkmark \\
		$S_4$ & & & & & & \checkmark \\
		\hline
		\hline
		\multirow{2}{*}{Score} & $0.31$  & $0.31$ & $0.31$
		&  $0.31$ & $0.31$ & \bm{$0.28$} \\
		&    $\pm  0.00 $ & $\pm 0.00 $ & $\pm 0.00 $
		&  $\pm 0.00$ & $\pm 0.00 $ & \bm{$\pm 0.00$}\\
		\hline
	\end{tabular}
\end{table}

\section{Additional Related Work}
\label{sec:liter}

Recent research has suggested that it is possible to enforce strong differential privacy protection in many types of analyses without significant utility loss (see~\cite{Dwork:2009:tcc} for an excellent survey).

The existing work can be roughly categorized into supervised settings, such as logistic regression~\cite{Chaudhuri:2011:erm} and
support vector machine~(SVM)~\cite{Chaudhuri:2011:erm, Rubinstein:2009:dpsvm},
and unsupervised settings, such as publishing histograms~\cite{Xu:2013:differentially},
releasing contingency tables~\cite{Yang:2012:differential},
 hypothesis testing~\cite{Gaboardi:2016:dpht},
 collaborative recommendation~\cite{Zhu:2016:dprecommend}, K-Means clustering~\cite{Su:2016:dpkmeans},
and spectral graph analysis~\cite{Wang:2013:dpga}. To our best knowledge, this work represents one of the first attempts in the direction of differentially private publishing of semantic-rich data.
%
%

More recently, extensive research effort has focused on enforcing differential privacy in training deep learning models. Adadi {\em et al.}  \cite{Abadi:2016:dpdl} proposed to use differentially private stochastic gradient descent~\cite{Song:2013:stochastic} to enforce $(\epsilon, \delta)$-differential privacy in training deep neural networks. Phan {\em et al.}~\cite{Phan:2016:dpae}, proposed to apply the
functional mechanism~\cite{Zhang:2012:functional} to train differentially private auto-encoders. In \cite{Phan:2017:dplap}, Phan {\em et al.} proposed an adaptive Laplace mechanism to reduce the required random noise. Our work advances this line of research by enforcing differential privacy in the setting of training generative adversarial networks, a new class of deep learning models.

The work most relevant to ours is perhaps~\cite{Gergely:2017:dpmixgen}, in which Gergely {\em et al.} proposed a framework of training differential private deep generative networks. Our work however differs from~\cite{Gergely:2017:dpmixgen} in significant ways.
%
First, \cite{Gergely:2017:dpmixgen} used a two-stage process that first performs clustering and then produces generative models such as
Restricted Boltzmann Machine (RBM)~\cite{Fischer:2012:rbmintro} and Variational Auto-Encoder (VAE)~\cite{Kingma:2013:vae}; in contrast, our work provides an end-to-end solution that produces general GAN, which is known to outperform RBM and VAE in data synthesis. Second, the method in \cite{Fischer:2012:rbmintro} only works well for low-dimensional data (e.g., 784 for MNIST and 1303 for CDR); in contrast, \system is able to generate high-quality, high-dimensional synthetic data (e.g., 12,288 for LSUN). In~\cite{Jones:2017:dpgan}, Jones {\em et al.} also proposed a differentially private GAN
framework, which however only generates low-dimensional samples~($ 3  \times 12 $) and meanwhile requires label information. In comparison,
\system works well for high-dimensional data without any labeling information. To achieve this, \system adopts multiple optimization strategies that improve both training stability and utility retention.

\section{Conclusion and Discussion}
\label{sec:end}

In this paper, we present \system, a generic framework of publishing semantic-rich data in a privacy-preserving manner. Instead of releasing sanitized datasets, \system releases differentially private generative models, which can be used by analysts to synthesize unlimited amount of data for arbitrary analysis tasks. To achieve this, \system integrates the generative adversarial network framework with differential privacy mechanisms, provides refined analysis of privacy loss within this framework, and employs a suite of optimization strategies to address the training stability and scalability challenges. Using benchmark datasets and analysis tasks, we show that \system is able to synthesize data of utility comparable to original data, at the cost of modest privacy loss.

This work also opens several avenues for further research. For example, in this paper we mostly focus on publishing image data, while it is worth investigation to adapt \system to support other types of semantic-rich data (e.g., LSTM for language modeling tasks). In addition, \system is formulated as an unsupervised framework, while its extension to supervised and semi-supervised learning is attractive for data with label information.

\bibliographystyle{acm}
\bibliography{main}



\end{document}